\documentclass[hidelinks,onefignum,onetabnum]{siamart251104}

\usepackage{amsmath,amssymb,amsfonts}
\usepackage{mathtools}
\usepackage{algorithmic}
\usepackage{graphicx}
\usepackage{tikz}
\usetikzlibrary{positioning,chains,fit,shapes,calc}
\usepackage{subcaption}

\usepackage{cleveref}

\usepackage{comment}
\usepackage{tabularx}
\usepackage{booktabs}
\usepackage{enumerate}
\usepackage{xcolor}
\usepackage{tcolorbox}
\usepackage{xspace} 

\usepackage{hyperref}
\hypersetup{
    colorlinks=true,
    citecolor=red,
    filecolor=blue,
    linkcolor=blue,
    urlcolor=blue
}

\DeclareMathOperator*{\E}{\mathsf{E}}

\newcommand{\sig}{\sigma}
\newcommand{\ep}{\epsilon}

\newcommand{\alp}{\alpha}
\newcommand{\Del}{\Delta}
\newcommand{\del}{\delta}

\newcommand{\sbp}[1]{\big(#1\big)}
\newcommand{\bp}[1]{\Big(#1\Big)}
\newcommand{\bP}[1]{\left(#1\right)}

\newcommand{\sbb}[1]{\big[#1\big]}
\newcommand{\bb}[1]{\Big[#1\Big]}
\newcommand{\bB}[1]{\left[#1\right]}

\newcommand{\cM}{\mathcal{M}}

\newcommand{\cU}{\mathcal{U}}

 \newcommand{\LP}{\ensuremath{\mathsf{LP}}\xspace}

\newcommand{\xhdr}[1]{\vspace{1mm} \noindent{\textbf{#1}}}

 \newcommand{\OPT}{\ensuremath{\mathsf{OPT}}\xspace}
\newcommand{\opt}{\ensuremath{\mathsf{OPT}}\xspace}
\def \dalg {\ensuremath{\mathsf{D\mbox{-}ALG}}\xspace}
\newcommand{\kvv}{\ensuremath{\mathsf{KVV}}\xspace}

\def \alg {\ensuremath{\mathsf{ALG}}\xspace}

\def \rad {\ensuremath{\mathsf{RGA}}\xspace}

\newcommand{\tT}{\tilde{T}}
\newcommand{\hT}{T}

\newcommand{\bluee}[1]{{\color{blue}#1}}
\newcommand{\redd}[1]{{\color{red}#1}}

\newcommand{\ie}{\emph{i.e.,}\xspace}
\newcommand{\eg}{\emph{e.g.,}\xspace}

\newcommand{\anhai}[1]{{\color{black}#1}} 

\headers{Correcting the KVV Upper-Bound Analysis}{P.~Xu}

\title{Correcting the Foundational Analysis of 
Karp--Vazirani--Vazirani (STOC 1990):\\
A Rigorous Revision of the \mbox{$1 - \exp(-1)$} Upper Bound%
\thanks{This is the second version as of November 23, 2025.}
}
\author{Pan Xu%
\thanks{Email: \texttt{panxu0@gmail.com}.}
}

\begin{document}
\maketitle

\begin{abstract}
We revisit the classical analysis of Karp, Vazirani, and Vazirani (\kvv) (STOC~1990), which established the well-known upper bound of \(1 - 1/e\) as the limiting proportion of vertices that can be matched by any online procedure in a canonical bipartite structure. Although foundational, the original analysis contains several inaccuracies, including a fundamental technical gap in the treatment of the underlying discrete process. We give a transparent and fully rigorous reconstruction of the \kvv\ argument by reformulating the evolution of available neighbors as a discrete-time death process and deriving a sharp upper bound via a simple factor-revealing linear program that captures the correct recurrence structure. This yields a precise bound \(\lceil n(1 - 1/e) + 2 - 1/e \rceil\) on the expected number of matched vertices, refining the classical claim \(n(1 - 1/e) + o(n)\). Our goal is \emph{not} to optimize this upper bound, but to provide a mathematically sound and conceptually clean \emph{correction} of the classical \kvv\ analysis, while remaining faithful to its original combinatorial framework.
\end{abstract}




\section{Introduction}

Since the classical work of Karp, Vazirani, and Vazirani~\cite{kvv}, the online
bipartite matching problem has become a central object of study in theoretical
computer science and discrete optimization. Numerous variants have been
proposed; see the surveys~\cite{huang2024online,mehta2012online}. Although
motivated originally by adversarial-order online computation, the model has
since appeared in a wide range of discrete allocation settings, including
ride-hailing platforms~\cite{feng2023two,pavone2022online,aouad2022dynamic,lowalekar2020competitive},
assortment and recommendation systems~\cite{chen2024assortment,gong2022online,ma2020algorithms},
and online resource allocation~\cite{manshadi2023fair,stein2020advance}. 

For convenience, we refer to the paper of~\cite{kvv} simply as \kvv.
With more than a thousand citations on Google Scholar by the end of~2024,
\kvv{} established two fundamental results that characterize the asymptotic
performance of online matching.  For a bipartite graph with $n$ vertices on
each side, \kvv{} claims that:
\begin{itemize}
\item[(1)] \textbf{Lower Bound.} The Ranking algorithm achieves an expected
matching size of at least $(1-1/e)\,n + o(n)$.
\item[(2)] \textbf{Upper Bound.} No online algorithm can obtain an expected
matching size exceeding $(1-1/e)\,n + o(n)$.
\end{itemize}
Together, these show that the Ranking algorithm is asymptotically optimal.

However, the proofs in \kvv{} are technically involved, and parts of the
derivation—especially in the argument for the upper bound—lack precision.
This has been noted in prior work.\footnote{For example,
\cite{eden2021economics} observed that ``The analysis in the original paper
(of \kvv) was quite complicated (and imprecise in places).''}  Considerable
effort has since been devoted to providing simpler or more rigorous proofs of
the lower bound~\cite{eden2021economics,birnbaum2008line,devanur2013randomized}.
In contrast, comparatively little attention has been paid to the correctness and
structure of the upper-bound analysis.

In this paper, we revisit in detail the upper-bound argument of \kvv.  
We identify several technical issues—some minor, others fundamental—in the original
derivation, and we introduce a simple yet rigorous method that resolves these
difficulties.  Our goal is to clarify and strengthen the classical argument,
while remaining faithful to its original combinatorial structure.

\subsection{Review of the Approach Establishing the $1 - 1/e$ Upper Bound in~\kvv}

\xhdr{Online Matching under Adversarial Order.}
Consider a bipartite graph $G = (I, J, E)$, where $I$ and $J$ represent the sets
of offline and online agents, respectively. Each agent in $I$ is static, whereas
each online agent $j \in J$ is revealed sequentially (referred to as
``arriving"), together with its set of offline neighbors, denoted by $N_j$. Upon
the arrival of an online agent $j$, an algorithm (\alg) must make an immediate
and irrevocable decision, without knowledge of future arrivals, either to
reject $j$ or to match $j$ to one of its neighbors $i \in N_j$. In the latter
case, the offline agent $i$ becomes unavailable (assuming unit capacity),
resulting in a successful match. The objective is to design an algorithm~$\alg$
that maximizes the expected total number of matches (assuming unweighted
edges). Note that \alg has no prior knowledge of the graph structure~$G$,
which is determined by an oblivious adversary prior to observing any
information about~\alg.

The work of~\kvv applies Yao's lemma to establish the upper bound of
$1 - 1/e$ for any online algorithm. Specifically, they construct a collection of
online matching instances as follows.\footnote{Although we aim to present the
exact idea in~\kvv, we use a slightly different set of notations for ease of
presentation.}
Consider the basic module instance shown in~\Cref{fig:ws-ins}, denoted
by $G = (I, J, E)$, where $|I| = |J| = n$ and each online agent
$j \in \{1, 2, \ldots, n\} := [n]$ has a set of offline neighbors
$N_j = \{\,i \mid j \le i \le n\,\}$. Alternatively, the instance $G$ can be
represented by a unique $n \times n$ binary matrix, denoted by~$M$, such that
each row $i \in [n]$ corresponds to offline agent~$i$, while each column
$j \in [n]$ corresponds to online agent~$j$. In this representation,
$M(i,j) = 1$ if and only if $(i,j) \in E$, \ie $i \in N_j$.

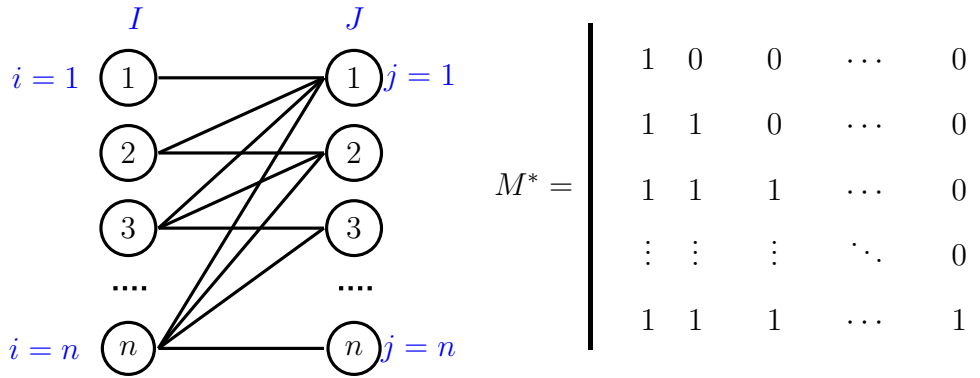
\begin{figure}[ht!]
\centering
\begin{minipage}{.5\linewidth}
\centering
\begin{tikzpicture}[scale=1]
 \draw (0.1,0.5) node[above] {\textcolor{blue}{$I$}};
 \draw (0,0) node[minimum size=0.2mm,draw,circle, very thick] {{$1$}};
 \draw (0,-1) node[minimum size=0.2mm,draw,circle, very thick] {{$2$}};
 \draw (0,-2) node[minimum size=0.2mm,draw,circle, very thick] {{$3$}};
 \draw[dotted, ultra thick] (-0.2,-2.8)--(0.3,-2.8);
 \draw (0,-3.6) node[minimum size=0.2mm,draw,circle, very thick] {{$n$}};

 \draw (3,0.5) node[above] {\textcolor{blue}{$J$}};
 \draw (3,0) node[minimum size=0.2mm,draw,circle, very thick] {{$1$}};
 \draw (3,-1) node[minimum size=0.2mm,draw,circle, very thick] {{$2$}};
 \draw (3,-2) node[minimum size=0.2mm,draw,circle, very thick] {{$3$}};
 \draw[dotted, ultra thick] (2.8,-2.8)--(3.3,-2.8);
 \draw (3,-3.6) node[minimum size=0.2mm,draw,circle, very thick] {{$n$}};

 \draw[-, very thick] (0.4,0)--(2.6,0);
 \draw[-, very thick] (0.4,-1)--(2.6,0);
 \draw[-, very thick] (0.4,-2)--(2.6,0);
 \draw[-, very thick] (0.4,-3.6)--(2.6,0);
 \draw[-, very thick] (0.4,-1)--(2.6,-1);
 \draw[-, very thick] (0.4,-2)--(2.6,-1);
 \draw[-, very thick] (0.4,-3.6)--(2.6,-1);
 \draw[-, very thick] (0.4,-2)--(2.6,-2);
 \draw[-, very thick] (0.4,-3.6)--(2.6,-2);
 \draw[-, very thick] (0.4,-3.6)--(2.6,-3.6);

 \draw (-0.5,0) node[left] {\bluee{$i = 1$}};
 \draw (-0.5,-3.6) node[left] {\bluee{$i = n$}};
 \draw (4.8,0) node[left] {\bluee{$j = 1$}};
 \draw (4.8,-3.6) node[left] {\bluee{$j = n$}};
\end{tikzpicture}
\end{minipage}%
\hfill
\begin{minipage}{.5\linewidth}
\renewcommand{\arraystretch}{1.7}
\centering
$
M^* =
\left|
\begin{array}{ccccccccc}
1 & 0 &   & 0 &   & \cdots &   & 0 \\
1 & 1 &   & 0 &   & \cdots &   & 0 \\
1 & 1 &   & 1 &   & \cdots &   & 0 \\
\vdots & \vdots &   & \vdots &   & \ddots &   & 0 \\
1 & 1 &   & 1 &   & \cdots &   & 1
\end{array}
\right |
$
\end{minipage}
\medskip
\caption{A basic module instance proposed by~\kvv{} to prove the upper bound of
\(1 - 1/e\). The instance (Left) can be represented using a unique \(n \times n\)
binary matrix \(M^*\) (Right), where \(M^*(i,j) = 1\) if and only if the online
agent \(j\) is connected to the offline agent \(i\), \ie \((i,j) \in E\).}
\label{fig:ws-ins}
\end{figure}

Let \(\cM\) denote the collection of all matrices generated by permuting the
\(n\) rows of \(M^*\) in~\Cref{fig:ws-ins}. Each matrix in \(\cM\) uniquely
determines an online matching instance, and vice versa; thus \(\cM\) may be
viewed as a collection of \(n!\) such instances. \emph{Throughout this paper we
refer to an online matching instance simply as a matrix in \(\cM\)}. Let
\(\cU(\cM)\) denote the uniform distribution over~\(\cM\).

\begin{lemma}[\cite{kvv}]
For any deterministic algorithm \(\dalg\), its expected performance (the
expected size of the matching obtained) on \(\cU(\cM)\), denoted \\
\(\E_{M \sim \cU(\cM)}[\sig(\dalg,M)]\), satisfies
\[
\E_{M \sim \cU(\cM)}[\sig(\dalg,M)]
=
\E[\sig(\rad,M^*)],
\]
where \(\rad\) is formally defined in~\Cref{alg:rad} as the Randomized Greedy
Algorithm that uniformly samples an available offline neighbor, if any, for
each arriving online agent. Here, \(\E[\sig(\rad,M^*)]\) represents its expected
performance on the fixed matrix \(M^*\), with the expectation taken over the
randomness of~\(\rad\).
\end{lemma}

By Yao's lemma, we obtain for any (possibly randomized) algorithm~\(\alg\),
\begin{align}\label{ineq:1-20-a}
\sig(\alg,\cM)
:=
\min_{M \in \cM} \E[\sig(\alg,M)]
&\le
\max_{\dalg}\;
\E_{M \sim \cU(\cM)}[\sig(\dalg,M)]
\\
&=
\E[\sig(\rad,M^*)], \nonumber
\end{align}
where the expectation \(\E[\sig(\alg,M)]\) is taken over the randomness in~\alg,
and \\ \(\max_{\dalg}\) ranges over all deterministic algorithms. For any instance
\(M \in \cM\), the offline optimum (the clairvoyant maximum matching), denoted
\(\opt\), always returns a matching of size~\(n\), i.e., \(\sig(\opt,M)=n\).
Therefore, Inequality~\eqref{ineq:1-20-a} implies that the competitiveness
achieved by any online algorithm over instances in \(\cM\) is upper bounded by
\(\E[\sig(\rad,M^*)]/n\). Consequently, it remains to characterize the exact or
asymptotic value of \(\E[\sig(\rad,M^*)]/n\) as \(n \to \infty\).

\subsection{Review of the Approach of Computing the Value \(\E\sbb{\sig(\rad, M^*)}\) in~\kvv}\label{sec:kvv}

For completeness, we formally state the randomized greedy (\rad) in~\Cref{alg:rad}. In this section, we rephrase the approach for computing the value  \(\E\sbb{\sig(\rad, M^*)}\) presented by \kvv, using a different set of notations and analysis, although we aim to recover the essential idea from their work.\footnote{The analysis in~\kvv\ is rather brief and lacks rigor, exhibiting both minor and major issues (see the remarks at the end of this section). We adopt a different set of notations and supply several missing technical details, which we believe help illuminate the major issues, as outlined in Section~\ref{sec:bug}.}

\begin{algorithm}
\caption{Randomized Greedy Algorithm (\rad)}
\label{alg:rad}
\begin{algorithmic}
\FOR{$j = 1,2,\ldots,n$}
    \STATE{Uniformly sample an available offline neighbor of $j$, if any, and match it to $j$}
\ENDFOR
\RETURN the set of all matched edges
\end{algorithmic}
\end{algorithm}

Consider applying \(\rad\) to the instance $M^*$ in~\Cref{fig:ws-ins}. For ease of presentation, we use \(t \in [n]\) to index the online agent \(j\) arriving at time \(t = j\), referring to it simply as online agent \(t\). For each time \(t \in [n]\), let \(Y_t\) represent the number of available offline neighbors of agent \(t\) before implementing any online actions from \(\rad\), with \(Y_1 = n\). For each $t \in [n]$, define \(\Del_t := Y_{t+1} - Y_t\) as the change in the number of available offline neighbors during time \(t\).

According to \kvv, \(\Del_t \in \{-1, -2\}\), and \(\Del_t = -2\) if and only if, when agent \(t\) arrives at time \(t\), the following two events occur simultaneously: 
(1) its offline neighbor \(i = t\) happens to be available, denoted by \(E_{t,1}\), and 
(2) some available offline neighbor of agent \(t\) other than \(i = t\) is sampled by \rad, denoted by \(E_{t,2}\). 

Consequently, \kvv\ claimed that
\begin{align*}
\Pr[E_{t,1} \mid Y_t] &= \frac{Y_t}{n-t+1}, \quad \Pr[E_{t,2} \mid Y_t, \, E_{t,1}] = 1 - \frac{1}{Y_t},
\end{align*}
which leads to \(\Pr[E_{t,1} \wedge E_{t,2} \mid Y_t] = \frac{Y_t - 1}{n-t+1}\). Thus, they further claimed that
\begin{align}
&\E[\Del_t \mid Y_t]= (-2) \cdot \frac{Y_t - 1}{n-t+1} + (-1) \cdot \bp{1 - \frac{Y_t - 1}{n-t+1}} \\
&\quad \quad \quad \quad= (-1) - \frac{Y_t - 1}{n-t+1}, \label{eqn:22-b}\\
\Rightarrow &\E[\Del_t] = (-1) - \frac{\E[Y_t] - 1}{n-t+1}. \label{eqn:21-c} \footnotemark
\end{align}
\footnotetext{Note that our notations $Y_t$ and $\Del_t$ correspond to $y(t)$ and $\Del y$ in \kvv. As a result, Equality~\eqref{eqn:21-c} corresponds to the equation $\frac{\E[\Del y]}{\E[\Del x]} = 1 + \frac{y(t) - 1}{x(t)}$ in \kvv, where, in their context, $x(t) = n - t + 1$ and $\Del x = -1$.}


For each \(t \in [n]\), define \(\mu_t := \E[Y_t]\) with \(\mu_1 = n\). Therefore, the above equality suggests that 
\begin{align}\label{eqn:21-b}
\mu_{t+1} - \mu_t = (-1) - \frac{\mu_t - 1}{n - t + 1}, \quad \mu_1 = n.
\end{align}

Define \(g(t/n) := \mu_t / n\), where \(g: [0,1] \to [0,1]\) is a continuous function with \(g(0) = 1\). The sequence \(\{\mu_t \mid t \in [n]\}\) can be approximated using an ordinary differential equation (ODE) as \(n \to \infty\), as follows:
\begin{align}
& \quad \mu_{t+1} - \mu_t = (-1) - \frac{\mu_t - 1}{n - t + 1} \\
\Leftrightarrow & \quad \frac{\mu_{t+1} - \mu_t}{n} = (-1) \cdot \frac{1}{n} \cdot \Big(1 + \frac{\mu_t / n - 1 / n}{1 - t / n + 1 / n}\Big) \label{eqn:22-a} \\
\Leftrightarrow & \quad g((t+1)/n) - g(t/n) = (-1) \cdot \frac{1}{n} \cdot \Big(1 + \frac{g(t/n) - 1/n}{1 - t/n + 1/n}\Big) \nonumber \\
\Leftrightarrow & \quad \frac{g((t+1)/n) - g(t/n)}{1/n} = (-1) \cdot \Big(1 + \frac{g(t/n) - 1/n}{1 - t/n + 1/n}\Big) \nonumber \\
\Rightarrow &  \quad \frac{\mathrm{d}g(z)}{\mathrm{d}z} = (-1) \cdot \Big(1 + \frac{g(z)}{1 - z}\Big) \label{eqn:21-d} \\ 
&\quad \textbf{(by setting } z = t/n \textbf{ and \(n \to \infty\))} \nonumber
\end{align}

Solving the ODE in~\eqref{eqn:21-d} with the boundary condition \(g(0) = 1\), we obtain:
\begin{align}
g(z) = (1 - z) \cdot \Big(1 + \ln(1 - z)\Big).\label{eqn:21-a} \footnotemark
\end{align}

\footnotetext{This corresponds to the equation $y = 1 + x \sbp{\frac{n-1}{n} \redd{-} \ln \frac{x}{n}}$ in \kvv, where the symbol $\redd{-}$ was a typo: it should be $\bluee{+}$.
 The paper of \kvv did not normalize $x(t) = n-t+1$ and $y(t)$ (corresponding to $\mu_t$ here) by scaling them down to $[0,1]$, as we did here. Thus, after dividing both sides of $y = 1 + x \sbp{\frac{n-1}{n} + \ln \frac{x}{n}}$ by $n$ and then taking $n \to \infty$, it reduces to  the exact equation in~\eqref{eqn:21-a} with $z := t/n$.}

Thus, by solving \(g(z) = 0\), we find a unique solution \(z^* = 1 - 1/e\). This indicates that when \(t^* = n \cdot z^* = n \cdot (1 - 1/e)\), we have \(g(z^*) = 0\), which further implies \(\mu_{t^*}/n \to 0\) as \(n \to \infty\). Equivalently, \(\mu_{t^*} = o(n)\).  As a result, \kvv concluded that the expected size of the matching returned by \(\rad\) is:
\begin{align*}
\E\big[\sig(\rad, M^*)\big] = t^* + o(n) = n \cdot (1 - 1/e) + o(n).
\end{align*}

\xhdr{Clarifying Remarks on the Above Rephrased Analysis of~\kvv.}  As mentioned earlier, our \anhai{notations} and analysis are slightly different from those in \kvv, although they are fundamentally based on the same idea. Specifically, they treated \(\mu_t\) (as opposed to the notation \(y(t)\) used in their work) as a function of \(x(t) := n - t + 1\). They developed an ODE involving \(\mathrm{d}y/\mathrm{d}x\), which takes a different form from ours, though the two are mathematically equivalent. 
Additionally, they did not normalize the values of \(y(t)\) and \(x(t)\) by scaling them down to \([0,1]\), as we did here, but instead directly set the initial values to \(x(1) = y(1) = n\). 

The derivation of the ODE involving \(g(z)\), specifically the steps from~\eqref{eqn:22-a} to~\eqref{eqn:21-d}, is omitted in~\kvv. For completeness, we supply the full derivation here.  We believe our approach of deriving the ODE is \emph{more rigorous} than that in \kvv, even though both approaches are conceptually the same. That being said, both approaches still suffer from a few fundamental technical flaws, as outlined in the next Section~\ref{sec:bug}.

\subsection{Fundamental Errors in the Approach of Computing \(\E\sbb{\sig(\rad, M^*)}\) in~\kvv}\label{sec:bug}

First, the approach of \kvv, as detailed in Section~\ref{sec:kvv}, essentially
computes the expectation of the time, denoted by \(\tT\), at which
\(\mu_{\tT} = o(n)\). \emph{However, our actual objective is to compute the
expectation of the time, denoted by \(T\), at which \(Y_{T+1}\) becomes zero};
see Equation~\eqref{def:T} for a precise definition. Importantly, there is no
clear relationship between these two expectation values.

Second, Equality~\eqref{eqn:22-b} implicitly conditions on the event
\(Y_t \ge 2\), rather than considering arbitrary values of \(Y_t\). This leads to
inaccuracies in subsequent Equations~\eqref{eqn:21-c} and~\eqref{eqn:21-b}.\footnote{More
precisely, Equality~\eqref{eqn:21-c} is valid when \(Y_t \ge 1\), but incorrect
when \(Y_t = 0\). A corrected version of Equation~\eqref{eqn:21-b} is provided
in Equality~\eqref{eqn:28-a}.} Consequently, this mistake invalidates the
resulting ODE in~\eqref{eqn:21-d} and its solution presented
in~\eqref{eqn:21-a}. In particular, we observe that when \(z > 1 - 1/e\), the
solution satisfies \(g(z) < 0\), implying a negative expected number of
available offline neighbors for online agents arriving at times
\(t > n \cdot (1 - 1/e)\), which is evidently impossible.

Furthermore, neglecting boundary conditions such as \(Y_t \ge 0\) fundamentally
\emph{invalidates} the analytical approach, even if the goal were only to
upper-bound \(\E[\tT]\) (the expectation of the time \(\tT\) defined
by \(\mu_{\tT} = o(n)\)). The dynamics expressed by
Equations~\eqref{eqn:21-c} and~\eqref{eqn:21-b} permit \(Y_t\) to continue
decreasing even after reaching zero. This oversight causes the expectations
\(\mu_t := \E[Y_t]\) and the normalized function \(g(z) = \mu_{z \cdot n}/n\)
(derived from the ODE~\eqref{eqn:21-d}) to be systematically smaller than their
correct counterparts. Specifically, the current expectations incorrectly
include scenarios in which \(Y_t\) becomes negative; such scenarios must instead
be treated as having value \(0\).

As a result, the non-increasing sequence \((\mu_t)\) declines at an
artificially accelerated rate compared to the actual process, causing the
turning point \(\tT\) (defined by \(\mu_{\tT} = o(n)\)) to occur prematurely.
Consequently, the solution \(z^*\) obtained from the incorrect equation
\(g(z) = 0\) (equivalently \(\mu_{z \cdot n} = o(n)\)) provides only a
\emph{lower bound} on the true value, even though the intended purpose of the
analysis is to derive an \emph{upper} bound.

\subsection{Main Contributions}

In this paper, we present a detailed and comprehensive review of the classical
analyses establishing the well-known \(1 - 1/e\) upper bound on the competitiveness
achievable by any online algorithm, originally introduced in~\kvv. We carefully
identify and correct both minor inaccuracies and fundamental errors in their
analytical approach. Our main contribution is a \emph{simple yet rigorous}
method that resolves these issues. Specifically, we establish the following
theorem:

\begin{theorem}\label{thm:main}
\[
\E\sbb{\sig(\rad, M^*)}
   \le \big\lceil (1-1/e) \cdot n + 2 - 1/e \big\rceil,
   \qquad \forall n = 1, 2, \ldots.
\]
\end{theorem}

\xhdr{Remarks on Theorem~\ref{thm:main}.}
(1) \kvv claimed that \(\E\sbb{\sig(\rad, M^*)} = (1 - 1/e)\, n + o(n)\).
Theorem~\ref{thm:main} refines this statement by showing that the \(o(n)\) term
is in fact at most a constant \(3 - 1/e\), yielding a sharper upper bound on the
competitiveness of any online algorithm.

(2) We emphasize that optimizing the constant \(2 - 1/e\) in
Theorem~\ref{thm:main} is not our primary objective. Our central goal is to
\emph{provide a simple, rigorous correction of the technically flawed proof
in~\kvv} and thereby \emph{rescue} the classical analytical framework developed
there. Our intent is not to discard the \kvv{} approach, but rather to repair
and clarify it—ensuring that this influential method rests on a sound technical
foundation—rather than to seek the tightest possible upper bound on
\(\E\sbb{\sig(\rad, M^*)}\).\footnote{For tighter upper bounds, readers may refer to~\cite{feige2020tighter}, 
which obtained the best-known upper bound 
\((1 - 1/e)\,n + 1 - 2/e + O(1/n!)\) for \(\E\sbb{\sig(\rad, M^*)}\). 
Their method relies on a combinatorial approach previously used to analyze online 
\emph{fractional} matching algorithms. 
Note, however, that the fractional-matching-based analytical framework in 
\cite{feige2020tighter} differs fundamentally from that of \kvv{} and from the 
present work, and is substantially more intricate and technically involved.}


\section{Formal Approach to Upper Bounding the Value $\E\sbb{\sig(\rad, M^*)}$}\label{sec:formal}
Recall that for each $t \in [n]$, $Y_t$ denotes the number of available offline neighbors of the agent arriving at time $t$ before any online actions from $\rad$, and $\Del_t:=Y_{t+1}-Y_t$ represents the decrease in the number of available offline neighbors during time $t$.
Below, we formulate the sequence $\{Y_t \mid t\in [n]\}$ as a pure discrete-time death  process. 
\begin{definition}\label{def:death}
\textbf{A Discrete-Time Death Process of $\{Y_t \mid t\in [n]\}$}: Consider the following discrete-time death process that starts at the state $Y_1=n$, and at any time $1 \le t \le n$, the transition probability is stated as follows: 
\begin{align*}
\Pr \bb{\Del_t=-2 \mid Y_t \ge 1}&=\frac{Y_t-1}{n-t+1},  \Pr \bb{\Del_t=-1 \mid Y_t \ge 1}=1-\frac{Y_t-1}{n-t+1}; \footnotemark\\
\Pr \bb{\Del_t=0 \mid Y_t =0}&=1.
\end{align*} 
\end{definition}
\footnotetext{The transition probabilities are derived directly from the case when \(Y_t \geq 2\); see the analysis of \(\Del_t\) in Section~\ref{sec:kvv} (which is indeed correct when \(Y_t \geq 2\)). We can verify that they also apply to the case when \(Y_t = 1\), in which \(\Del_t = -1\) surely.}

By the above definition, the death process \(\{Y_t \mid t \in [n]\}\) surely reaches the death state of \(0\) by the end of \(n\) rounds, i.e., \(\Pr[Y_{n+1} = 0] = 1\),\footnote{Here, \(Y_{n+1}\) can be interpreted simply as the number of available offline neighbors of agent \(t = n\) \emph{at the end} of \(\rad\).} since at least one death occurs during time \(t\) whenever \(Y_t \geq 1\).

Define:
\begin{align}
\hT := \min_{1 \leq t \leq n} \{t \mid Y_{t+1} = 0\},   \label{def:T} 
\end{align}
where \(\hT\) represents the death time. By the nature of \(\rad\), we claim that \(\hT\) is the exact (\emph{random}) number of matches obtained by \(\rad\), \ie, \(\hT = \sig(\rad, M^*)\).

\subsection{Upper Bounding the Value \(\E[\hT]\) with \(\hT\) Defined in~\eqref{def:T}}


Below, we present the correct form of Equation~\eqref{eqn:21-b} regarding the dynamics of the expectation of \(Y_t\).

\begin{lemma}\label{lem:26-a}
For each $t \in [n]$, let  \(\alp_t = \Pr[Y_t \geq 1]\) and  \(\del_t = \E[Y_t]\) with $\del_1=n$, where \(\{Y_t\}\) is the death process defined in~\eqref{def:death}.\footnote{Here, we use a distinct notation, \(\del_t\), to clearly differentiate it from \(\mu_t\), which denotes the expectation of \(Y_t\) as analyzed in Section~\ref{sec:kvv}. As discussed in Section~\ref{sec:bug}, Equation~\eqref{eqn:21-b}, which involves \(\{\mu_t\}\), is \emph{incorrect} due to overlooking boundary cases (specifically, when \(Y_t = 0\)).}
  We have that
\begin{align}\label{eqn:28-a}
\del_1=n;  \quad \del_{t+1}= \bp{1- \frac{1}{n-t+1}} \cdot \bp{\del_t-\alp_{t}}, \forall t \in [n].\footnotemark
\end{align}
\end{lemma}

\footnotetext{Observe that when $t=n$, $\del_{n+1}=0$, which aligns with the fact that $Y_{n+1}=0$ surely.}
\begin{proof}
By the definition of $\{Y_t\}$, we have
\begin{align*}
\E[\Del_t] &=\Pr[Y_t \ge 1] \cdot \E[\Del_t \mid Y_t \ge 1]\\
&= \Pr[Y_t \ge 1] \cdot (-2) \cdot \frac{\E[Y_t \mid Y_t \ge 1]-1}{n-t+1}\\
&\quad +\Pr[Y_t \ge 1] \cdot (-1) \cdot \bp{1-\frac{\E[Y_t \mid Y_t \ge 1]-1}{n-t+1}}\\
&=(-1) \cdot  \Pr[Y_t \ge 1] \cdot \bp{1+\frac{\E[Y_t \mid Y_t \ge 1]-1}{n-t+1}}\\&=
(-1) \cdot \bp{ \Pr[Y_t \ge 1] + \frac{\E[Y_t]-\Pr[Y_t \ge 1]}{n-t+1}}
\end{align*}

Observe that $\alp_t=\Pr[Y_t \ge 1]$ and $\E[\Del_t]=\del_{t+1}-\del_t$. Substituting these two to the above equality, we have
\begin{align*}
&\del_{t+1}-\del_t=(-1) \cdot \bp{ \alp_t+\frac{\del_t-\alp_t}{n-t+1}} \\
\Leftrightarrow &~~
\del_{t+1}=\del_t- \alp_t-\frac{\del_t-\alp_t}{n-t+1}=(\del_t- \alp_t) \cdot \bp{1- \frac{1}{n-t+1}}.
\end{align*}
\end{proof}

\begin{lemma}\label{lem:25-a}
For each $t \in [n]$, $
\Pr[\hT \ge t]=\Pr[Y_t \ge 1]$.
\end{lemma}
\begin{proof}
Consider a given \(t \in [n]\). It suffices to show that the two events, \((\hT \geq t)\) and \((Y_t \geq 1)\), are equivalent. Observe that \((T \geq t)\) means that the death process reaches the final death state after at least \(t\) rounds, which implies that the number of available offline neighbors of agent \(t\) must be at least \(1\) upon its arrival, \ie \(Y_t \geq 1\). Conversely, the event \( (Y_t \geq 1)\) suggests that the death process has not yet reached the death state after \(t - 1\) rounds, which is precisely \((\hT \geq t)\).
\end{proof}
Recall that $\alp_t=\Pr[Y_t \ge 1]$. By Lemma~\ref{lem:25-a}, we have
\begin{align*}
\E[\hT]=\sum_{t=1}^n \Pr[\hT \ge t]=\sum_{t=1}^n \Pr[Y_t \ge 1]=\sum_{t=1}^n  \alp_t.
\end{align*}
We formulate the task of upper bounding the value \(\E[\hT]\) as the following maximization linear program (LP):
\begin{align}
\max_{(\alp_t) \in [0,1]^n} &\E[\hT] = \sum_{t \in [n]} \alp_t \label{lp:25-1} \\
& \del_{t+1} = \bp{1 - \frac{1}{n-t+1}} \cdot \bp{\del_t - \alp_t}, &&\forall t \in [n], \label{eqn:26-a} \\
& 0 \leq \alp_t \leq \min(\del_t, 1), &&\forall t \in [n], \label{eqn:25-a} \\
& \del_1 = n. \nonumber
\end{align}
For notational convenience, we refer to the above LP as \LP~\eqref{lp:25-1}. Note that Constraint~\eqref{eqn:26-a} is due to Lemma~\ref{lem:26-a}, while Constraint~\eqref{eqn:25-a} follows from the fact that \(\alp_t = \Pr[Y_t \geq 1] \leq \E[Y_t] = \del_t\) by Markov's inequality.\footnote{Observe that we could impose additional constraints on \((\alp_t)\), \eg \(\alp_t = 1\) for all \(1 \leq t \leq \lfloor n/2 \rfloor\). This follows directly from the definition of the Death Process \((Y_t)\) in~\eqref{def:death}, as at most two deaths can occur in each round. Consequently, the death process cannot terminate within \(\lfloor n/2 \rfloor-1\) rounds, ensuring \(\Pr[T \geq \lfloor n/2 \rfloor] = 1 = \Pr[Y_{\lfloor n/2 \rfloor} \geq 1] = \alp_{\lfloor n/2 \rfloor}\). Another natural constraint is \(\alp_t \geq \alp_{t+1}\) for all \(t \in [n]\), since \((Y_t)\) is non-increasing. Nevertheless, we do not explicitly add these constraints to \LP~\eqref{lp:25-1}, as the current constraints already guarantee that any optimal solution satisfies them automatically; see Lemma~\ref{lem:25-c} for details. Thus, explicitly adding these constraints does not affect the structure of the optimal solution.}

\begin{lemma}\label{lem:25-c}
For any optimal solution \((\alp_t)\) to \LP~\eqref{lp:25-1}, we have \(\alp_t = \min(\del_t, 1)\) for all \(t \in [n]\). 
\end{lemma}

\begin{proof}
Assume, by contradiction, that there exists some \(\ell \in [n]\) such that \(\alp_{\ell} < \min(\del_{\ell}, 1)\). Consider the following two cases:

\textbf{Case 1}: \(\ell = n\). We claim that \((\alp_t)\) is surely not an optimal solution since we can safely increase \(\alp_n\) to \(\min(\del_{\ell}, 1)\), which immediately leads to a strictly better solution (in terms of a larger objective value).

\textbf{Case 2}: \(\ell < n\). We can construct another solution \((\alp'_t)\) based on \((\alp_t)\) as follows:
\[
\alp'_t =
\begin{cases} 
\alp_t, & \text{if } t < \ell \text{ or } t > \ell + 1, \\
\alp_\ell + \ep, & \text{if } t = \ell, \\
\alp_{\ell+1} - \ep \cdot \bp{1 - \frac{1}{n-\ell+1}}, & \text{if } t = \ell+1,
\end{cases}
\]
where 
\[
\ep = \min\bP{\min(\del_\ell, 1) - \alp_\ell, ~~\frac{\alp_{\ell+1}}{1 - 1/(n-\ell+1)}}.
\]
Let \((\del_t)\) and \((\del'_t)\) be the sequences associated with \((\alp_t)\) and \((\alp'_t)\), respectively. We can verify the following two things.
First, \((\alp'_t)\) is feasible. Note that \(0 < \ep \leq \min(\del_\ell, 1) - \alp_\ell\) and \(\alp'_\ell = \alp_\ell + \ep\). Thus, \(0 < \alp'_\ell \leq \min(\del_\ell, 1) = \min(\del'_\ell, 1)\) since \(\del_t = \del'_t\) for all \(t \leq \ell\). This justifies the feasibility of Constraint~\eqref{eqn:25-a} on \(\alp'_\ell\). Additionally:
\begin{align*}
\del'_{\ell+1} &= (\del'_\ell - \alp'_\ell) \cdot \bp{1 - \frac{1}{n-\ell+1}} = (\del_\ell - \alp_\ell - \ep) \cdot \bp{1 - \frac{1}{n-\ell+1}}, \\
\del'_{\ell+2} &= (\del'_{\ell+1} - \alp'_{\ell+1}) \cdot \bp{1 - \frac{1}{n-\ell}} \\
&= \bB{(\del_\ell - \alp_\ell - \ep) \cdot \bp{1 - \frac{1}{n-\ell+1}} - \alp_{\ell+1} + \ep \cdot \bp{1 - \frac{1}{n-\ell+1}}} \\
& \quad \cdot \bp{1 - \frac{1}{n-\ell}} \\
&= \bB{(\del_\ell - \alp_\ell) \cdot \bp{1 - \frac{1}{n-\ell+1}} - \alp_{\ell+1}} \cdot \bp{1 - \frac{1}{n-\ell}} \\
&= \bp{\del_{\ell+1} - \alp_{\ell+1}} \cdot \bp{1 - \frac{1}{n-\ell}} = \del_{\ell+2}.
\end{align*}
The above analysis shows that \(\del'_{\ell+1} - \alp'_{\ell+1} = \del_{\ell+1} - \alp_{\ell+1} \geq 0\), which implies \(\alp'_{\ell+1} \leq \del'_{\ell+1}\). Note that \(\alp'_{\ell+1} < \alp_{\ell+1} \leq 1\). Thus, \(\alp'_{\ell+1} \leq \min(\del'_{\ell+1}, 1)\). By the definition of \(\ep\), we also have \(\alp'_{\ell+1} \geq 0\). Together, this establishes the feasibility of Constraint~\eqref{eqn:25-a} on \(\alp'_{\ell+1}\). Since \(\del'_{\ell+2} = \del_{\ell+2}\) and \(\alp'_t = \alp_t\) for all \(t < \ell\) and \(t > \ell+1\), the feasibility of Constraint~\eqref{eqn:25-a} for \(\alp'_t\) at all \(t < \ell\) and \(t > \ell+1\) follows from the feasibility of \((\alp_t)\).

Second, \((\alp'_t)\) is strictly better than \((\alp_t)\) since:
\begin{align*}
\sum_{t \in [n]} \alp'_t &= \sum_{t < \ell} \alp'_t + \sum_{t > \ell+1} \alp'_t + \alp'_{\ell} + \alp'_{\ell+1} \\
&= \sum_{t < \ell} \alp_t + \sum_{t > \ell+1} \alp_t + \alp_{\ell} + \ep + \alp_{\ell+1} - \ep \cdot \bp{1 - \frac{1}{n-\ell+1}} \\
&= \sum_{t \in [n]} \alp_t + \frac{\ep}{n-\ell+1} > \sum_{t \in [n]} \alp_t.
\end{align*}

Thus, we conclude that \((\alp'_t)\) is a feasible solution strictly better than \((\alp_t)\), which contradicts the optimality of \((\alp_t)\).
\end{proof}

\subsection{Proof of the Main Theorem~\ref{thm:main}}
Observe that Lemma~\ref{lem:25-c} provides a simple way to pinpoint the exact optimal solution to \LP~\eqref{lp:25-1}. Specifically, the optimal solution greedily assigns each $\alp_t$ as large as possible until $\del = 0$. To facilitate our analysis, we propose an auxiliary sequence as follows:\footnote{Our idea here shares the essence of the coupling technique in analyzing the convergence rate of Markov chains.}
\begin{align}\label{seq:26-a}
\del^*_1 = n, \quad \del^*_{t+1} = \bp{1 - \frac{1}{n-t+1}} \cdot \sbp{\del^*_t - 1}, \forall t \in [n].
\end{align}
We can solve the above sequence analytically as follows:
\begin{align}\label{series:26-a}
\del^*_t &= (n+1-t) \cdot \bp{1 + H_{n+1-t} - H_n} \nonumber \\&= (n+1-t) \cdot \bp{1 - \sum_{\ell=n+2-t}^n \frac{1}{\ell}}, \quad \forall t \in [n],
\end{align}
where $H_k := \sum_{\ell=1}^k \frac{1}{\ell}$ denotes the $k$-th harmonic number; see a formal proof in~\ref{app:harr}.

\begin{proof}[Proof of Theorem~\ref{thm:main}]
Let $\OPT(\LP)$ denote the optimal value of~\LP~\eqref{lp:25-1}. It suffices to show that $\OPT(\LP) \le \lceil (1-1/e) \cdot n + 2 - 1/e \rceil$, since $\E[\sig(\rad, M^*)]=\E[T] \le \OPT(\LP)$ due to the validity of~\LP~\eqref{lp:25-1}.

Observe that the solution in~\eqref{series:26-a} satisfies that for each $t \in [n]$,
\begin{align*}
\del^*_t &= (n+1-t) \cdot \bp{1-\sum_{\ell=n+2-t}^n \frac{1}{\ell}} \le (n+1-t) \cdot \bp{1- \int_{n+2-t}^{n+1} \frac{1}{x} ~\mathrm{d}x}\\
& = (n+1-t) \cdot \bp{1-\ln \bp{\frac{n+1}{n+2-t}}} := g(t).
\end{align*}

We can verify that $\hat{t} = (1-1/e) \cdot n + 2 - 1/e$ is the unique solution (though possibly fractional) to the equation $g(t) = 0$ over the continuous range $[1, n]$. Let $t^* = \lceil ~\hat{t}~ \rceil = \lceil (1-1/e) \cdot n + 2 - 1/e \rceil$. Since both $(n+1-t)$ and $\sbp{1-\ln \sbp{\frac{n+1}{n+2-t}}}$ are decreasing over $t \ge 1$, we claim that 
\[
\del^*_{t^*} \le g(t^*) \le g(\hat{t}) = 0.
\]
Observe that $t^* = \lceil (1-1/e) \cdot n + 2 - 1/e \rceil$ is an integer. 

\textbf{Case 1}: $\del^*_{t^*} = 0$. Then we have $\OPT(\LP) = t^*$ by Lemma~\ref{lem:25-c}.
 
\textbf{Case 2}: $\del^*_{t^*} < 0$. Then we have $\del^*_{t^*-1} \in (0, 1)$. By Lemma~\ref{lem:25-c}, we have 
\[
\OPT(\LP) = t^* - 1 + \del^*_{t^*-1} < t^*.
\]

Summarizing the above two cases, we conclude that $\OPT(\LP) \le t^*$. Thus, we establish the claim. 
\end{proof}

\section{Conclusion}\label{sec:con}
In this paper, we revisited the well-known $1 - 1/e$ upper bound on the competitiveness of online algorithms, as originally established in the classical work of \kvv. Through a detailed review of their analysis, we identified and addressed several minor and major technical issues in their approach. To resolve these issues, we proposed a \emph{simple yet rigorous} method, which provides a more precise upper bound of $\lceil (1 - 1/e) \cdot n + 2 - 1/e \rceil$ on the performance of any online algorithm, replacing the original $(1 - 1/e) \cdot n + o(n)$ bound. Our approach is not only notable for its simplicity but also significantly less technically involved than existing methods. By improving both the rigor and clarity of the analysis, we aim to make the foundational results on the competitiveness of online algorithms more accessible and robust for future research.

\appendix

\section{Solving the Sequence $(\del^*_t)_{t \in [n]}$ as Defined in~\eqref{seq:26-a}}\label{app:harr}

Recall that the auxiliary sequence $(\del^*_t)_{t \in [n]}$ defined in~\eqref{seq:26-a} is stated as follows:
\begin{align*}
\del^*_1 = n, \quad \del^*_{t+1} = \bp{1 - \frac{1}{n-t+1}} \cdot \sbp{\del^*_t - 1}, \quad \forall t \in [n].
\end{align*}

\begin{lemma}
The analytical formula of $(\del^*_t)$ in~\eqref{series:26-a} is a valid solution to the sequence defined in~\eqref{seq:26-a}.
\end{lemma}

\begin{proof}
We prove this by induction over $t = 1, 2, \ldots, n$. 

For $t = 1$, we can verify that $\del^*_1 = n$, which conforms to the definition. Assume that for any $1 \leq k \leq t$, we have:
\begin{align*}
\del^*_k = (n+1-k) \cdot \bp{1 - \sum_{\ell=n+2-k}^n \frac{1}{\ell}}, \quad \forall 1 \leq k \leq t.
\end{align*}
Now, we prove the case for $t = k+1$:
\begin{align*}
\del^*_{k+1} &= \bp{1 - \frac{1}{n-k+1}} \cdot \sbp{\del^*_k - 1} \\
&= \frac{n-k}{n-k+1} \cdot \bB{(n+1-k) \cdot \bp{1 - \sum_{\ell=n+2-k}^n \frac{1}{\ell}} - 1} \\
&= (n-k) \cdot \bp{1 - \sum_{\ell=n+2-k}^n \frac{1}{\ell}} - \frac{n-k}{n-k+1} \\
&= (n-k) \cdot \bp{1 - \sum_{\ell=n+2-k}^n \frac{1}{\ell} - \frac{1}{n-k+1}} \\
&= (n-k) \cdot \bp{1 - \sum_{\ell=n+1-k}^n \frac{1}{\ell}}.
\end{align*}

Thus, we complete the inductive proof for the case $t = k+1$. 
\end{proof}

\bibliographystyle{siamplain}
\bibliography{all}  

\end{document}